\documentclass[usenames,dvipsnames,12pt]{article}

\usepackage{graphicx}
\usepackage{amsmath}
\usepackage{amssymb}
\usepackage{latexsym}
\usepackage{mathrsfs}
\usepackage{amsthm}
\usepackage{setspace}
\usepackage{epsfig}
\usepackage{subfigure}
\usepackage{authblk}
\usepackage{amsfonts}
\usepackage{url}
\usepackage{float}
\usepackage{multibib}
\usepackage{titling}
\usepackage[authoryear,round]{natbib}	
\usepackage{mathtools}														 % need for `show only references'
\mathtoolsset{showonlyrefs=true}									 % only equations which are labeled AND referenced will be numbered.
\usepackage{geometry}
\usepackage[dvips]{color}
\newtheoremstyle{newplain}
{4pt}
{4pt}
{\itshape}
{}
{\itshape\bf}
{.}
{.5em}
{}
\theoremstyle{newplain}
\newtheorem{theorem}{Theorem}[section]
\newtheorem{proposition}{Proposition}[section]

\newtheorem{definition}{Definition}[section]
\newtheorem{assumption}{Assumption}[section]

\newtheoremstyle{newdefinition}
{4pt}
{4pt}
{}
{}
{\itshape\bf}
{.}
{.5em}
{}
\theoremstyle{newdefinition}

\newtheorem{vg}{Example}[section]

\newtheorem{remark}{Remark}[section]
\numberwithin{equation}{section}
\newcommand{\be}{\begin{equation}}
\newcommand{\ee}{\end{equation}}
\newcommand{\nee}{\nonumber\end{equation}}
\newcommand{\eel}[1]{\label{#1}\end{equation}}
\newcommand{\brmk}[1]{\begin{remark}\label{#1}\begin{em} }
\newcommand{\ermk}{ $\quad\triangleleft$\end{em}\end{remark}}
\newcommand{\bvg}[1]{\begin{vg}\label{#1}\begin{em} }
\newcommand{\evg}{ $\quad\triangleleft$\end{em}\end{vg}}

\begin{document}
\bibliographystyle{plainnat}
\setlength{\abovedisplayskip}{8pt}
\setlength{\belowdisplayskip}{8pt}
\setlength{\abovedisplayshortskip}{4pt}
\setlength{\belowdisplayshortskip}{8pt}

\begin{titlepage}

\title{{\bf{\Huge Long-Term Factorization of Affine Pricing Kernels}}
\thanks{This paper is based on research supported by the grant CMMI-1536503 from the National Science Foundation.}}
\author{Likuan Qin\thanks{likuanqin2012@u.northwestern.edu}
}
\author{Vadim Linetsky\thanks{linetsky@iems.northwestern.edu}
}
\affil{\emph{Department of Industrial Engineering and Management Sciences}\\
\emph{McCormick School of Engineering and Applied Sciences}\\
\emph{Northwestern University}}
\date{}
\end{titlepage}

\maketitle

\begin{abstract}
This paper constructs and studies the long-term factorization of affine pricing kernels into discounting at the rate of return on the long bond and the martingale component that accomplishes the change of probability measure to the long forward measure. The principal eigenfunction of the affine pricing kernel germane to the long-term factorization is an exponential-affine function of the state vector with the coefficient vector identified with the fixed point of the Riccati ODE. The long bond volatility and the volatility of the martingale component are explicitly identified in terms of this fixed point. A range of examples from the asset pricing literature is provided to illustrate the theory.
%
%\keywords{Stochastic discount factor \and Affine models \and Long-term factorization \and Long forward measure \and Principal eigenfunction}
%\PACS{E43}
\end{abstract}

\newpage

\section{Introduction}

The stochastic discount factor (SDF) is a fundamental object in arbitrage-free asset pricing models. It assigns today's prices to risky future payoffs at alternative investment horizons. It accomplishes this by simultaneously discounting the future and adjusting for risk.
A familiar representation of the SDF is a factorization into discounting at the risk-free interest rate and a martingale component adjusting for risk. This martingale accomplishes the change of probabilities to the risk-neutral probability measure.
More recently \citet{alvarez_2005using}, \citet{hansen_2008consumption}, \citet{hansen_2009} and \citet{hansen_2012} introduce and study an alternative {\em long-term factorization} of the SDF. The {\em transitory component} in the long-term factorization discounts at the rate of return on the pure discount bond of asymptotically long maturity (the {\em long bond}). The {\em permanent component} is a  martingale that accomplishes a change of probabilities to the {\em long forward measure}. \citet{linetsky_2014long} study the long-term factorization and the long forward measure in the general semimartingale setting.

The long-term factorization of the SDF is particularly convenient in applications to the pricing of long-lived assets  and to theoretical and empirical investigations of the term structure of the risk-return trade-off. In addition to the references above, the growing literature on the long-term factorization and its applications includes \citet{hansen_2012pricing},  \citet{hansen_2013}, \citet{borovicka_2014mis}, \citet{borovivcka2011risk}, \citet{borovivcka2016term},
\citet{bakshi_2012}, \citet{bakshia2015recovery}, \citet{christensen2014nonparametric},  \citet{christensen_2013estimating}, \citet{linetsky_2014_cont}, \citet{linetsky2016bond}, \citet{backus2015term}, \citet{filipovic2016linear}, \citet{filipovic2016relation}. Empirical investigations in this literature show that the martingale component in the long-term factorization is highly volatile and economically significant (see, in particular, \citet{bakshi_2012} for results based on pricing kernel bounds,  \citet{christensen2014nonparametric} for results based on structural asset pricing models connecting to the macro-economic fundamentals, and  \citet{linetsky2016bond} for results based on explicit parameterizations of the pricing kernel, where, in particular, the relationship among the measures ${\mathbb P}$, ${\mathbb Q}$ and ${\mathbb L}$ is empirically investigated).

The focus of the present paper is on the analysis of long-term factorization in affine diffusion models, both from the perspective of providing a user's guide to constructing long-term factorization in affine asset pricing models, as well as employing affine models as a convenient laboratory to illustrate the theory of the long-term factorization.
Affine diffusions are work-horse models in continuous-time finance due to their analytical and computational tractability (\citet{vasicek_1977equilibrium}, \citet{cox_1985_2}, \citet{duffie_1996}, \citet{duffie_2000}, \citet{dai_2000}, \citet{duffie_2003}). In this paper we show that  the principal eigenfunction of \citet{hansen_2009} that determines the long-term factorization, if it exists, is necessarily in the exponential-affine form in affine models, with the coefficient vector in the exponential identified with the fixed point of the corresponding Riccati ODE.
This allows us to give a fully explicit treatment and illustrate dynamics of the long bond, the martingale component and the long-forward measure in affine models. In particular, we explicitly verify that when the Riccati ODE associated with the affine pricing kernel possesses a fixed point, the affine model satisfies the sufficient condition in Theorem 3.1 of  \citet{linetsky_2014long} so that the long-term limit exists.

In Section \ref{Brownian} we review and summarize the long-term factorization in Brownian motion-based models. In Section \ref{exist_affine} we present general results on the long-term factorization of affine pricing kernels. The main results are given in Theorem \ref{affine_long}, where the market price of Brownian risk is explicitly decomposed into the market price of risk under the long forward measure identified with the volatility of the long bond and the remaining market price of risk determining the martingale component accomplishing the change of probabilities from the data-generating to the long forward measure. The latter component is determined by the fixed point of the Riccati ODE.
In Section \ref{examples} we study a range of examples of affine pricing kernels from the asset pricing literature.

\section{Long-Term Factorization in Brownian Environments}
\label{Brownian}

We work on a complete filtered probability space $(\Omega,{\mathscr F},({\mathscr F}_{t})_{t\geq 0},{\mathbb P})$.
We assume that all uncertainty in the economy is generated by an $n$-dimensional Brownian motion $W_t^{\mathbb{P}}$ and that
 $(\mathscr{F}_t)_{t\geq 0}$ is the (completed) filtration generated by $W_t^{\mathbb{P}}$. We assume absence of arbitrage and market frictions, so that there exists a strictly  positive pricing kernel process in the form of an It\^{o} semimartingale. More precisely, we assume that the pricing kernel follows an It\^{o} process ($\cdot$ denotes vector dot product)
$$
dS_t=-r_tS_tdt-S_t \lambda_t \cdot dW_t^{\mathbb{P}}
$$
with $\int_0^t |r_s|ds<\infty$ and the market price of Brownian risk vector $\lambda_t$ such that the process
$$
M_t^0=e^{-\int_0^t \lambda_s\cdot dW_s^{\mathbb{P}}-\frac{1}{2}\int_0^t \|\lambda_s\|^2 ds}
$$
is a martingale (Novikov's condition ${\mathbb E}^{\mathbb P}[e^{\frac{1}{2}\int_0^t \|\lambda_s\|^2 ds}]<\infty$ for each $t>0$ suffices).
Under these assumptions the pricing kernel has the risk-neutral factorization
\be
S_t=\frac{1}{A_t}M_t^0=e^{-\int_0^t r_s ds}M_t^0
\eel{rnf}
into discounting at the risk-free short rate
$r_t$ determining the risk-free asset (money market account) $A_t=e^{\int_0^t r_s ds}$ and the exponential martingale $M^0_t$ with the market price of Brownian risk $\lambda_t$ determining its volatility.
We also assume that ${\mathbb E}^{\mathbb{P}}[S_T/S_t]<\infty$ for all $T>t\geq 0$. The integrability of the SDF $S_T/S_t$ for any two dates $T>t$ ensures that that zero-coupon bond price processes
$$P_t^T:=\mathbb{E}_t^\mathbb{P}[S_T/S_t], \quad t\in [0,T]$$
are well defined for all maturity dates $T>0$ ($\mathbb{E}_t[\cdot]=\mathbb{E}[\cdot|{\mathscr F}_{t}]$).

Since for each $T$ the $T$-maturity zero coupon bond price process $P_t^T$ can be written as $P_t^T=M_t^TP_0^T/S_t$, where $M_t^T=S_t P_t^T/P_0^T= \mathbb{E}_t^\mathbb{P}[S_T]/\mathbb{E}_0^\mathbb{P}[S_T]$ is a positive martingale on $t\in [0,T]$, we can apply the  Martingale Representation Theorem to claim that $$dM_t^T=-M_t^T \lambda_t^T \cdot dW^{\mathbb P}_t$$ with some $\lambda_t^T$, and further claim that the bond price process has the representation
$$
dP_t^T=(r_t+ \sigma^T_t\cdot \lambda_t) P_t^Tdt+P_t^T \sigma^T_t \cdot dW_t^\mathbb{P}
$$
with the volatility process $\sigma^T_t=\lambda_t-\lambda_t^T$.

Following \citet{linetsky_2014long}, for each fixed $T>0$ we define a self-financing trading strategy that rolls over investments in $T$-maturity zero-coupon bonds as follows.
%For a given $T$, the $T$-forward measure is defined on   ${\mathscr F}_{T}$ (and, hence, on ${\mathscr F}_{t}$ for all $t\leq T$). We now extend it to ${\mathscr F}_{t}$ for all $t> T$ as follows.
Fix $T$ and consider a self-financing roll-over strategy that starts at time zero by investing one unit of account in $1/P_{0}^T$ units of the $T$-maturity zero-coupon bond. At time $T$ the bond matures, and the value of the strategy is $1/P_{0}^T$ units of account. We roll the proceeds over by re-investing into $1/(P_{0}^T P_{T}^{2T})$ units of the zero-coupon bond with maturity $2T$. We continue with the roll-over strategy, at each time $kT$ re-investing the proceeds into the bond $P_{kT}^{(k+1)T}$. We denote the valuation process of this self-financing strategy $B_t^T$:
\[
B_t^T = \left(\prod_{i=0}^k P_{iT}^{(i+1)T}\right)^{-1} P_{t}^{(k+1)T},\quad t\in [kT,(k+1)T),\quad k=0,1,\ldots.
\]
For each $T>0$, the process $B_t^T$ is defined for all $t\geq 0$.
The process $S_t B_t^T$ extends the martingale $M_t^T$ to all $t\geq 0$. It thus defines the $T$-{\em forward measure} ${\mathbb Q}^T|_{{\mathscr F}_{t}}=M_t^T {\mathbb P}|_{{\mathscr F}_{t}}$
 on  ${\mathscr F}_{t}$ for each $t\geq 0$, where $T$ now has the meaning of the length of the compounding interval. Under the $T$-forward measure ${\mathbb Q}^T$ extended to all ${\mathscr F}_{t}$, the roll-over strategy $(B_t^T)_{t\geq 0}$ with the compounding interval $T$ serves as the  numeraire asset. Following \citet{linetsky_2014long}, we continue to call the measure extended to all ${\mathscr F}_{t}$ for $t\geq 0$ the $T$-forward measure and use the same notation, as it reduces to the standard definition of the forward measure on ${\mathscr F}_{T}$.

Since the roll-over strategy $(B^T_t)_{t\geq 0}$ and the positive martingale $M_t^T=S_t B_t^T$ are defined for all $t\geq 0$, we can write the $T$-{\em forward  factorization} of the pricing kernel for all $t\geq 0$:
\be
S_t = \frac{1}{B_t^T}M_t^T.
\eel{Tfactorization}

We now recall the definitions of the {\em long bond} and the {\em long forward measure} from \citet{linetsky_2014long}.
\begin{definition}{\bf (Long Bond)}
\label{def_longbond}
If the wealth processes $(B^T_t)_{t\geq 0}$ of the roll-over strategies in $T$-maturity bonds converge to a strictly positive semimartingale $(B_t^\infty)_{t\geq 0}$ uniformly on compacts in probability
as $T\rightarrow \infty$, i.e. for all $t>0$ and $K>0$
\[
\lim_{T\rightarrow \infty} {\mathbb P}(\sup_{s\leq t}|B_s^T-B_s^\infty|>K)=0,
\]
we call the limit the {\em long bond}.
\end{definition}

\begin{definition}{\bf (Long Forward Measure)}
\label{def_longforward}
If there exists a measure $\mathbb{Q}^\infty$ equivalent to $\mathbb{P}$ on each ${\mathscr F}_t$ such that the $T$-forward measures converge strongly to ${\mathbb Q}^\infty$ on each ${\mathscr F}_t$, i.e.
\[
\lim_{T\rightarrow \infty}{\mathbb Q}^T(A)={\mathbb Q}^\infty(A)
\]
for each $A\in {\mathscr F}_t$ and each $t\geq 0$,
we call the limit the {\em long forward measure} and denote it ${\mathbb L}$.
\end{definition}
The following theorem, proved in
\citet{linetsky_2014long}, gives a sufficient condition that ensures convergence to the long bond in the semimartingale topology which is stronger than the ucp convergence in Definition 1 and convergence of $T$-forward measures to the long forward measure in total variation, which is stronger than the strong convergence in Definition 2 (we refer to  \citet{linetsky_2014long} and the on-line appendix for proofs and details).

\begin{theorem}{\bf (Long Term Factorization and the Long Forward Measure)}
\label{implication_L1}
Suppose that for each $t>0$ the ratio of the ${\mathscr F}_t$-conditional expectation of the pricing kernel $S_T$ to its unconditional expectation converges to a positive limit in $L^1$ as $T\rightarrow \infty$ (under ${\mathbb P}$), i.e. for each $t>0$ there exists an almost surely positive ${\mathscr F}_t$-measurable random variable which we denote $M_t^\infty$ such that
\be
\frac{{\mathbb E}^{\mathbb P}_t[S_T]}{{\mathbb E}^{\mathbb P}[S_T]} \xrightarrow{\rm L^1} M_t^\infty\quad \text{as} \quad T\rightarrow \infty.
\eel{PKL1}
Then the following results hold:\\
(i) The collection of random variables $(M_t^\infty)_{t\geq0}$ is a positive ${\mathbb P}$-martingale, and the family of martingales $(M_t^T)_{t\geq 0}$ converges to  the martingale $(M_t^\infty)_{t\geq0}$ in the semimartingale topology.\\
(ii) The long bond valuation process $(B_t^\infty)_{t\geq0}$ exists, and the roll-over strategies $(B_t^T)_{t\geq 0}$ converge to the long bond $(B_t^\infty)_{t\geq 0}$ in the semimartingale topology.\\
(iii) The pricing kernel possesses the long-term factorization
\be
S_t=\frac{1}{B_t^\infty}M_t^\infty.
\eel{ltf}
(iv) $T$-forward measures ${\mathbb Q}^T$ converge to the long forward measure ${\mathbb L}$ in total variation  on each ${\mathscr F}_t$, and ${\mathbb L}$ is equivalent to ${\mathbb P}$ on ${\mathscr F}_t$ with the Radon-Nikodym derivative $M_t^\infty$.
\end{theorem}

The process $B_t^\infty$ has the interpretation of the gross return earned starting from time zero up to
time $t$ on holding the zero-coupon bond of asymptotically
long maturity. The long bond is the numeraire asset under the long forward measure $\mathbb{L}$ since the pricing kernel becomes $1/B_t^\infty$ under $\mathbb{L}$.
The long-term factorization of the pricing kernel \eqref{ltf} decomposes it
into discounting at the rate of
return on the long bond and a martingale component encoding a further risk adjustment.

Suppose the condition \eqref{PKL1} in Theorem \ref{implication_L1} holds in the Brownian setting of this paper.
Then the long bond valuation process is an It\^{o} semimartingale with the representation
$$
dB_t^\infty=(r_t+ \sigma^\infty_t \cdot \lambda_t) B_t^\infty dt + B_t^\infty \sigma^\infty_t  \cdot dW_t^\mathbb{P}
$$
with some volatility process $\sigma^\infty_t$ such that the process $M_t^\infty=S_t B_t^\infty$ satisfying
$$
dM_t^\infty=-M_t^\infty \lambda_t^\infty   \cdot dW_t^\mathbb{P}
$$
with $\lambda_t^\infty=\lambda_t - \sigma^\infty_t$
is a martingale (the permanent component in the long-term factorization). Thus, the long-term factorization Eq.\eqref{ltf} in the Brownian setting yields a decomposition of the market price of Brownian risk
\be
\lambda_t=\sigma^\infty_t + \lambda_t^\infty
\eel{mprdecomposition}
into the volatility of the long bond $\sigma_t^\infty$ and the volatility $\lambda_t^\infty$ of the martingale  $M_t^\infty$. The change of probability measure from the data-generating measure ${\mathbb P}$ to the long forward measure ${\mathbb L}$ is accomplished via Girsanov's theorem  with the ${\mathbb L}$-Brownian motion
$W_t^\mathbb{L}=W_t^\mathbb{P}+ \int_0^t \lambda_s^\infty ds.$

\section{Long Term Factorization of Affine Pricing Kernels}
\label{exist_affine}

We assume that the underlying economy is described by a Markov process $X$. We further assume $X$ is an affine diffusion and the pricing kernel $S$ is exponential affine in $X$ and the time integral of $X$.
Affine diffusion models are widely used in continuous-time finance due to their analytical tractability (\citet{vasicek_1977equilibrium}, \citet{cox_1985_2}, \citet{duffie_1996}, \citet{duffie_2000}, \citet{dai_2000}, \citet{duffie_2003}).
We start with a brief summary of some of the key facts about affine diffusions. We refer the reader to \citet{filipovic_2009} for  details, proofs and references to the literature on affine diffusion.

The process we work with solves the following SDE on the state space $E=\mathbb{R}_+^m\times\mathbb{R}^n$ for some $m,n\geq 0$ with $m+n=d$, where
$\mathbb{R}_+^m=\big\{x\in \mathbb{R}^m : x_i\geq 0$ for $i=1,...,m\big\}$:
\be
d X_t=b(X_t)dt+\sigma(X_t)d W^{\mathbb{P}}_t,\quad X_0=x,
\eel{affinesde}
where $W^{\mathbb P}$ is a $d$-dimensional standard Brownian motion and the diffusion matrix $\alpha(x)=\sigma(x)\sigma(x)^\dagger$ (here $^\dagger$ denotes matrix transpose to differentiate it from superscript $^T$) and the drift vector $b(x)$ are both affine in $x$:
\be
\alpha(x)=a+\displaystyle{\sum_{i=1}^d}x_i\alpha_i,\quad
b(x)=b+\displaystyle{\sum_{i=1}^d}x_i\beta_i=b+Bx
\ee
for some $d\times d$-matrices $a$ and $\alpha_i$ and $d$-dimensional vectors $b$ and $\beta_i$, where we denote by $B=(\beta_1,...,\beta_d)$ the $d\times d$-matrix with $i$-th column vector $\beta_i$, $1\leq i\leq d$. The first $m$ coordinates of $X$ are CIR-type and are non-negative, while the last $n$ coordinates are OU-type.  Define the index sets
$\emph{I}=\{1,...,m\}$  and $\emph{J}=\{m+1,...,m+n\}$.
For any vector $\mu$ and matrix $\nu$, and index sets $\emph{M},\emph{N}\in \{I,J\}$, we denote by
$\mu_\emph{M}=(\mu_i)_{i\in \emph{M}},$ $\nu_{\emph{M}\emph{N}}=(\nu_{ij})_{i\in \emph{M},j\in \emph{N}}$
the respective sub-vector and sub-matrix. To ensure the process stays in the domain $E={\mathbb R}_+^m\times {\mathbb R}^n$, we need the following assumption (cf. \citet{filipovic_2009})
\begin{assumption}{\bf (Admissibility)}\\
(1) $a_{JJ}$ and $\alpha_{i,JJ}$ are symmetric positive semi-definite for all $i=1,2,...,m$,\\
(2) $a_{II}=0,$ $a_{IJ}=a_{JI}^\dagger=0$,\\
(3) $\alpha_j=0$ for $j\in J$,\\
(4) $\alpha_{i,kl}=\alpha_{i,lk}=0$ for $k\in I\backslash \{i\}$ for all $1\leq k,l\leq d,$\\
(5) $b_I\geq 0$, $B_{IJ}=0$, and $B_{II}$ has non-negative off-diagonal elements.
\label{admi_and_nonde}
\end{assumption}
The condition $b_I\geq 0$ on the constant term in the drift of the CIR-type components ensures that the process stays in the state space $E$.
Making a stronger assumption $b_I>0$ ensures that the process instantaneously reflects from the boundary $\partial E$ and re-enters the interior of the state space ${\rm int}E=\mathbb{R}_{++}^m\times\mathbb{R}^n,$ where $\mathbb{R}_{++}^m=\big\{x\in \mathbb{R}^m : x_i> 0$ for $i=1,...,m\big\}$.
For any parameters satisfying  Assumption \ref{admi_and_nonde}, there exists a unique strong solution of the SDE \eqref{affinesde} (cf. Theorem 8.1 of \citet{filipovic_2009}). Denote by ${\mathbb P}_x$ the law of the solution $X^x$ of the SDE \eqref{affinesde}  for $x\in E$, ${\mathbb P}_x(X_t\in A):={\mathbb P}(X^x_t\in A)$.
Then $P_t(x,A)={\mathbb P}_x(X_t\in A)$ defined for all $t\geq 0$, Borel subsets $A$ of $E$, and $x\in E$ defines a Markov transition semigroup $(P_t)_{t\geq 0}$  on the Banach space of Borel measurable bounded functions on $E$ by $P_tf(x):=\int_E f(y)P_t(x,dy)$. As shown in  \citet{duffie_2003}, this semigroup is {\em Feller}, i.e., it leaves the space of continuous functions vanishing at infinity invariant. Thus, the Markov process $((X_t)_{t\geq 0},({\mathbb P}_x)_{x\in E})$ is a {\em Feller process} on $E$. It has continuous paths in $E$ and has the strong Markov property (cf. \citet{yamada_1971}, Corollary 2, p.162). Thus, it is a Borel right process (in fact, a Hunt process).

We make the following assumption about the pricing kernel.
\begin{assumption}{\bf (Affine Pricing Kernel)}\label{assumption_affine_PK}
We assume that the pricing kernel is exponential-affine in $X$ and its time integral:
\be
S_t=e^{-\gamma t-u^\dagger (X_t-X_0)-\int_0^t \delta^\dagger X_s ds},
\eel{affine_pk}
where $\gamma$ is a scalar and $u$ and $\delta$ are $d$-vectors and $^\dagger$ denotes matrix transpose.
\end{assumption}
%\begin{remark}
%There are other specification of affine pricing kernel in the literature. For example, we can assume $S_t=e^{-\gamma t+ u^\dagger (X_t-X_0)-\int_0^t \delta^\dagger X_s ds+\int_0^t\sigma_1(X_s)d W^{\mathbb{p}}_s+\int_0^t \sigma_2(X_s)dB^{\mathbb{P}}_s}$ where $B^{\mathbb{P}}_s$ is an independent Brownian motion, $\mu(x)$ is an affine function and $\sigma_1(x)$ and $\sigma_2(x)$ satisfy some restrictions so that $S_t$ is affine. However, we can add one additional factor $Y_t:=\int_0^t\sigma_1(X_s)d W_s+\int_0^t \sigma_2(X_s)dB_s$ and re-write $S_t=e^{-\gamma t+u^\dagger (X_t-X_0)+(Y_t-Y_0)-\int_0^t \delta^\dagger X_sds}$.
%\end{remark}
The pricing kernel
in this form is a positive multiplicative functional of the Markov process $X$.
The associated pricing operator ${\mathscr P}_t$ is defined by
$$
{\mathscr P}_tf(x)={\mathbb E}^{\mathbb P}_x[S_t f(X_t)]
$$
for a payoff $f$ of the Markov state. We refer the reader to Qin and Linetsky (2016a) for a detailed treatment of Markovian pricing operators.
The pricing kernel in the form \eqref{affine_pk} is called affine due to the following key result that shows that the term structure of pure discount bond yields is affine in the state vector $X$ (cf. \citet{filipovic_2009} Theorem 4.1).
\begin{proposition}
\label{affine_ZCB}
Let $T_0>0$. The following statements are equivalent: \\
(i) ${\mathbb E}^{\mathbb P}[S_{T_0}]<\infty$ for all fixed initial states $X_0=x\in {\mathbb R}_+^m\times {\mathbb R}^n$. \\
(ii) There exists a unique solution  $(\Phi(\cdot),\Psi(\cdot)):[0,T_0]\rightarrow {\mathbb R}\times {\mathbb R}^d$ of the following Riccati system of equations up to time $T_0$:
\be
\begin{split}
&\Phi^\prime(t)=-\frac{1}{2}\Psi_J(t)^\dagger a_{JJ}\Psi_J(t)+b^\dagger\Psi(t)+\gamma, \quad \Phi(0)=0,\\
&\Psi_i^\prime(t)=-\frac{1}{2}\Psi(t)^\dagger \alpha_{i}\Psi(t)+\beta_i^\dagger\Psi(t)+\delta_i,\quad i\in\emph{I},\\
&\Psi_J^\prime(t)=B_{JJ}^\dagger\Psi_J(t)+\delta_J,\quad \Psi(0)=u.\\
\end{split}
\eel{riccati_d}
In either case, the pure discount bond valuation processes (with unit payoffs) are exponential-affine in $X$:
\be
P_t^T=\mathbb{E}^{\mathbb P}_t[ S_T/S_t]= ({\mathscr P}_{T-t}1)(x)= P(T-t,X_t)=e^{-\Phi(T-t)-(\Psi(T-t)-u)^\dagger X^x_t}
\eel{representation}
for all $0\leq t\leq T\leq t+T_0$ and the SDE initial condition $x\in {\mathbb R}_+^m\times {\mathbb R}^n$.
\end{proposition}
Since in this paper our standing assumption is that ${\mathbb E}^{\mathbb P}[S_t]<\infty$ for all $t$, in this case the Riccati ODE system has solutions $\Psi(t)$ and $\Phi(t)$ for all $t$, and the bond pricing function entering the expression \eqref{representation} for the zero-coupon bond process
\be
P(t,x)=({\mathscr P}_t1)(x)=e^{-\Phi(t)-(\Psi(t)-u)^\dagger x}
\eel{bondfunction}
is defined for all $t\geq 0$ and $x\in E$.

We next show that an affine pricing kernel always possesses the risk-neutral factorization  with the affine short rate function.
\begin{theorem}{\bf (Risk-Neutral Factorization of Affine Pricing Kernels)}\label{RN_affine}
Suppose $X$ satisfies Assumption \ref{admi_and_nonde} and the pricing kernel satisfies Assumption \ref{assumption_affine_PK} together with the assumption that ${\mathbb E}^{\mathbb P}_x[S_t]<\infty$ for all $t\geq 0$ and every fixed initial state $X_0=x\in {\mathbb R}_+^m\times {\mathbb R}^n$.\\
(i) Then the pricing kernel admits the risk-neutral factorization
$$S_t=e^{-\int_0^t r(X_s)ds}M^0_t$$
with the affine short rate
\be
r(x)=g+h^\dagger x,\,
\eel{affineshortr}
with
\be
g=\gamma-\frac{1}{2}u_J^\dagger a_{JJ} u_J+ b^\dagger u, \, h_i=\delta_i-\frac{1}{2}u^\dagger\alpha_i u+\beta_i^\dagger u,\, i\in I, \, h_J=\delta_J+B_{JJ}^\dagger u_J
\eel{gh}
and the martingale
$$M^0_t=e^{-\int_0^t\lambda_s^\dagger dW_s^{\mathbb P}-\frac{1}{2}\int_0^t \|\lambda_s\|^2ds}$$
with the market price of Brownian risk (column $d$-vector)
\be
\lambda_t = \sigma(X_t)^\dagger u,
\eel{mprvaffine}
where $\sigma(x)$ is the volatility matrix of the state variable $X$ in the SDE \eqref{affinesde} and $$\|\lambda_t\|^2=\lambda_t^\dagger\lambda_t=u^\dagger \alpha(X_t)u.$$ \\
(ii) Under the risk-neutral measure ${\mathbb Q}$ defined by the martingale $M$, the dynamics of $X$ reads
\be
d X_t=(b(X_t)-\alpha(X_t)u) dt+\sigma(X_t)d W^{\mathbb{Q}}_t,
\eel{affineq}
where $W^{\mathbb Q}_t=W_t^{\mathbb P} +\int_0^t \lambda_s ds$
is the standard Brownian motion under ${\mathbb Q}$.
\end{theorem}
\begin{proof}
(i) Define a process $M_t^0:=S_te^{\int_0^t r(X_s)ds}$. It is also in the form of Eq.\eqref{affine_pk} with $\gamma$ replaced by $\gamma-g$ and $\delta$ replaced by $\delta-h$. Thus, Proposition \ref{affine_ZCB} also holds if we replace $S_t$ with $M_t^0$,  replace $\gamma$ with $\gamma-g$ and replace $\delta$ with $\delta-h$, i.e. $\mathbb{E}_t^\mathbb{P}[M_T/M_t]=e^{-\Phi(T-t)-(\Psi(T-t)-u)^\dagger X^x_t},$ where
\be
\begin{split}
&\Phi^\prime(t)=-\frac{1}{2}\Psi_J(t)^\dagger a_{JJ}\Psi_J(t)+b^\dagger\Psi(t)+\gamma-g, \quad \Phi(0)=0,\\
&\Psi_i^\prime(t)=-\frac{1}{2}\Psi(t)^\dagger \alpha_{i}\Psi(t)+\beta_i^\dagger\Psi(t)+\delta_i-h_i,\quad i\in\emph{I},\\
&\Psi_J^\prime(t)=B_{JJ}^\dagger\Psi_J(t)+\delta_J-h_J,\quad \Psi(0)=u.\\
\end{split}
\ee
With the choice of $g$ and $h$ in Eq.\eqref{gh}, the solution to the above ODE is $\Phi(t)=0$ and $\Psi(0)=u$, which implies $\mathbb{E}_t^\mathbb{P}[M_T/M_t]=1$. This shows that $M_t^0$ is a martingale. Furthermore, using the SDE for the affine state $X$, we can cast $M_t^0$ in the exponential martingale form $e^{-\int_0^t \lambda_s^\dagger dW_s^{\mathbb{P}}-\frac{1}{2}\int_0^t \|\lambda_s\|^2ds}$. with $\lambda_t$ given in \eqref{mprvaffine}.

\noindent(ii) The SDE for $X$ under $\mathbb{Q}$ follows from Girsanov's Theorem. $\Box$
\end{proof}

We next turn to the long term factorization of the affine pricing kernel.
\begin{theorem}{\bf (Long Term Factorization of Affine Pricing Kernels)}
\label{affine_long}
Suppose the solution $\Psi(t)$ of the Riccati ODE \eqref{riccati_d} converges to a fixed point $v\in {\mathbb R}^d$:
\be
\lim_{t\rightarrow\infty}\Psi(t)=v.
\eel{psi_converge}
Then the following results hold.\\
(i) Condition Eq.\eqref{PKL1} is satisfied and, hence, all results in Theorem \ref{implication_L1} hold. \\
(ii) The long bond is given by
\be
B_t^\infty=e^{\lambda t}\frac{\pi(X_t)}{\pi(X_0)},
\eel{long_bond_affine}
where
\be
\pi(x)=e^{(u-v)^\dagger x}
\eel{affineeigen}
is the positive exponential-affine eigenfunction of the pricing operator
${\mathscr P}_t$
$$
{\mathscr P}_t \pi(x)=e^{-\lambda t}\pi(x)
$$
with the eigenvalue $e^{-\lambda t}$ with
\be
\lambda=\gamma-\frac{1}{2}v_J^\dagger a_{JJ}v_J+ b^\dagger v
\eel{affineeigenv}
interpreted as the limiting long-term zero-coupon yield:
\be
\lim_{t\rightarrow \infty}\frac{-\ln P(t,x)}{t}=\lambda
\eel{asymptyield}
for all $x$.\\
(iii) The long bond has the ${\mathbb P}$-measure dynamics:
\be
dB_t^\infty = (r(X_t)+(\sigma^\infty_t)^\dagger \lambda_t )B_t^\infty dt + B_t^\infty (\sigma_t^\infty)^\dagger  dW_t^{\mathbb P},
\ee
where the (column vector) volatility of the long bond  is given by:
\be
\sigma_t^\infty=\sigma(X_t)^\dagger(u-v).
\eel{sigmainf}
(iv) The martingale component in the long-term factorization of the PK  $M^\infty_t=S_t B_t^\infty$ can be written in the form
\be
M^\infty_t=e^{-\int_0^t (\lambda_s^\infty)^\dagger dW_s^{\mathbb P}-\frac{1}{2}\int_0^t \|\lambda_s^\infty\|^2 ds},
\eel{minfty}
where
\be
\lambda^\infty_t=\lambda_t - \sigma_t^\infty=\sigma(X_t)^\dagger v.
\eel{gammainfty}
(v) The long-term decomposition of the market price of Brownian risk
is given by:
\be
\lambda_t = \sigma_t^\infty + \lambda^\infty_t,
\ee
where $\sigma_t^\infty$ is the volatility of the long bond \eqref{sigmainf} and $\lambda_t^\infty$ given in \eqref{gammainfty} defines the martingale \eqref{minfty}.\\
(vi) Under the long forward measure $\mathbb{L}$ the state vector $X_t$ solves the following SDE
\be
dX_t=(b(X_t)-\alpha(X_t)v)dt+\sigma(X_t)dW_t^\mathbb{L},
\eel{affinel}
where
$W_t^\mathbb{L}=W_t^\mathbb{P}+\int_0^t \lambda_s^\infty ds$
is the d-dimensional Brownian motion under $\mathbb{L}$, and the long bond has the ${\mathbb L}$-measure dynamics:
\be
dB_t^\infty = (r(X_t)+\|\sigma_s^\infty\|^2 )B_t^\infty dt + B_t^\infty (\sigma_t^\infty)^\dagger  dW_t^{\mathbb L}.
\ee
\end{theorem}
\begin{proof}
Since the solution of the Riccati ODE $\Psi(t)$ converges to a constant as $t\rightarrow \infty$, the right hand side of Eq.\eqref{riccati_d} also converges to a constant. This implies that $\Psi'(t)$ also converges to a constant. This constant must vanish, otherwise $\Psi(t)$ cannot converge to a constant. Thus, the right hand side of Eq.\eqref{riccati_d} also converges to zero. All these imply that $\Psi(t)=v$ is a stationary solution of the Riccati equation Eq.\eqref{riccati_d}.
Applying Proposition \ref{affine_ZCB} to the affine kernel of the form $1/B_t^\infty$, where $B_t^\infty$ is the process defined in \eqref{long_bond_affine}, it then follows that $\pi(x)$ defined in Eq.\eqref{affineeigen} is an eigenfunction of the pricing operator with the eigenvalue \eqref{affineeigenv}.
We can then verify that
$$M_t^\infty := S_t e^{\lambda t}\frac{\pi(X_t)}{\pi(X_0)}$$
is a martingale (with $M_0^\infty=1$).
We can use it to define a new probability measure
$$\mathbb{Q}^\pi|_{\mathscr{F}_t}:=M_t^\infty\mathbb{P}|_{\mathscr{F}_t}$$
associated with the eigenfunction $\pi(x)$.
The dynamics of $X_t$ under $\mathbb{Q}^\pi$ follows from Girsanov's Theorem. We stress that $\pi(x)$ is the eigenfunction of the pricing semigroup operator, rather than merely an eigenfunction of the generator. It is generally possible for an eigenfunction of the generator to fail to be an eigenfunction of the semigroup. That case will lead to a mere local martingale. In our case, $\pi(x)$ is an eigenfunction of the semigroup by construction, and the process $M_t^\infty$ is a martingale, rather than a mere local martingale.

We now show that the condition \eqref{PKL1} holds under our assumptions in Theorem 3.2.
We first re-write it under the probability measure $\mathbb{Q}^{\pi}$:
\be
\lim_{T\rightarrow\infty}\mathbb{E}^{\mathbb{Q}^\pi}\left[\left|\frac{P_t^T}{P_0^TB_t^\infty}-1\right|\right]=0.
\eel{BHS_L1}
We will now verify that this indeed holds under our assumptions.
First observe that by Eq.\eqref{representation}:
\be
\frac{P_t^T}{P_0^TB_t^\infty}=e^{-\lambda t-(\Phi(T-t)-\phi(T))-(\Psi(T-t)-v)^\dagger (X_t- X_0)}.
\ee
Since $\lim_{T\rightarrow\infty}\Psi(T)=v$ and $\lim_{T\rightarrow\infty}\Phi^\prime(T)=\lambda$, we have that
$$\lim_{T\rightarrow\infty}\frac{P_t^T}{P_0^TB_t^\infty}=1$$ almost surely.
Next, we show $L^1$ convergence. First, we observe that for any $\epsilon>0$ there exists $T_0$ such that for all $T>T_0$ $$|\Psi_i(T-t)-v_i|\leq\epsilon$$ for all $i\in I$ and $$e^{-\lambda t-(\Phi(T-t)-\phi(T))+(\Psi(T)-v)^\dagger X_0}\leq 1+\epsilon.$$
Thus, $$\left|\frac{P_t^T}{P_0^TB_t^\infty}-1\right|\leq1+\left|\frac{P_t^T}{P_0^TB_t^\infty}\right|\leq1+(1+\epsilon)\sum_{k_i=\pm\epsilon}e^{k^\dagger X_t}.$$
Since $X_t$ remains affine under $\mathbb{Q}^\pi$, by Theorem 4.1 of \citet{filipovic_2009} there exists $\epsilon>0$ such that $e^{k^\dagger X_t}$ is integrable under $\mathbb{Q}^\pi$ for all vectors $k$ such that $k_i=\pm\epsilon$. Thus, by the Dominated Convergence Theorem, Eq.\eqref{BHS_L1} holds. This proves (i) and (ii) (Eq.\eqref{asymptyield} follows from Eq.\eqref{bondfunction} and the fact $\Phi'(t)\rightarrow\lambda$ as $t\rightarrow\infty$).  (iii) follows from Eq.\eqref{long_bond_affine} and Ito's formula.

To prove (iv), we note that by Theorem \ref{implication_L1} $M_t^\infty$ is a  martingale. By It\^{o}'s formula, its volatility is $-\lambda_t^\infty$. This proves (iv). Part (v) follows from Eq.\eqref{gammainfty}. To prove (vi), first note that Eq.\eqref{minfty} and Girsanov's theorem implies that $W_t^\mathbb{L}=W_t^\mathbb{P}+\int_0^t\lambda_s^\infty ds$ is an $\mathbb{L}$-Brownian motion. The dynamics of $X_t$ and $B_t^\infty$ under $\mathbb{L}$ then follows. $\Box$
\end{proof}

The economic meaning of Theorem 3 is that the existence of a fixed point $v$ of the solution to the Riccati equation is sufficient for existence of the long term limit. The fixed point $v$ itself identifies the volatility of the long bond in Eq.(17) and the long-term zero-coupon yield in Eq.(16) via the principal eigenvalue (15).

We note that the condition in Theorem 3.2 of \citet{linetsky_2014long} is automatically satisfied in affine models. Indeed, from Eq.\eqref{bondfunction} when the Riccati equation has a fixed point $v$, from Theorem 3.2 in this paper we have
$$
\lim_{T\rightarrow \infty}\frac{P(T-t,x)}{P(T,x)}=e^{\lambda t},
$$
and we can write $P(t,x)=e^{-\lambda t}L_x(t)$, where $L_x(t)=e^{\lambda t}P(t,x)$ is a slowly varying function of time $t$ for each $x$. By Eq.\eqref{asymptyield}, the eigenvalue $\lambda$ is identified with the asymptotic long-term zero-coupon yield.

We note that since $\Psi(t)=v$ is a stationary solution of the Riccati ODE \eqref{riccati_d}, the vector $v$ satisfies the following {\em quadratic vector equation}:
$$\frac{1}{2}v^\dagger \alpha_{i}v+\beta_i^\dagger v-\delta_i=0,\quad i\in\emph{I},\quad B_{JJ}^\dagger v_J-\delta_J=0.$$
However, in general this quadratic vector equation may have multiple solutions leading to multiple exponential-affine eigenfunctions.
In order to determine the solution that defines the long-term factorization, if it exists, it is essential to verify that $v$ is the limiting solution of the Riccatti ODE, i.e. that Eq.\eqref{psi_converge} holds.
In this regard, we recall that \citet{linetsky_2014_cont} identified the unique {\em recurrent eigenfunction} $\pi_R$ of an affine pricing kernel with the {\em minimal} solution of the quadratic vector equation (see Appendix F in the on-line e-companion to \citet{linetsky_2014_cont}).
We recall that, for a Markovian pricing kernel $S$ (see \citet{hansen_2009} and \citet{linetsky_2014_cont}), we can associate a martingale
$$M^\pi_t=S_t e^{\lambda t}\frac{\pi(X_t)}{\pi(X_0)}$$ with {\em any} positive eigenfunction $\pi(x)$. In general, positive eigenfunctions are not unique. \citet{linetsky_2014_cont} proved uniqueness of a recurrent eigenfunction $\pi_R$ defined as such a positive eigenfunction  of the pricing kernel $S$, i.e. $${\mathbb E}_x^{\mathbb P}[S_t \pi(X_t)]=e^{-\lambda t}\pi(x)$$ for some $\lambda$,  that, under the locally equivalent probability measure (eigen-measure) ${\mathbb Q}^{\pi_R}$ defined by using the associated martingale $M_t^{\pi_R}$ as the Radon-Nikodym derivative, the Markov state process $X$ is recurrent.
However, in general, without additional assumptions, the recurrent eigenfunction $\pi_R$ associated with the minimal solution to the quadratic vector equation may or may not coincide with the eigenfunction $\pi_L$ germane to the long-term limit and, thus, the long forward measure may or may not coincide with the recurrent eigenmeasure (the fixed point $v$ of the Riccati ODE may or may not be the minimal solution of the quadratic vector equation).
Under additional exponential ergodicity assumptions the fixed point of the Riccati ODE is necessarily the minimal solution of the quadratic vector equation and  $\pi_R=\pi_L$. If the exponential ergodicity assumption is not satisfied, they may differ, or one may exist, while the other does not exist. We refer the reader to \citet{linetsky_2014_cont} and \citet{linetsky_2014long} for the exponential ergodicity assumption. Analytical tractability of affine models allows us to provide fully explicit examples to illustrate these theoretical possibilities. In the next section we give a range of examples.

\section{Examples}
\label{examples}

\subsection{Cox-Ingersoll-Ross Model}
\label{example_cir}

Suppose the state follows a CIR diffusion (\citet{cox_1985_2}):
\be
dX_t=(a -\kappa_{\mathbb P} X_t)dt+\sigma\sqrt{X_t}dW^{\mathbb{P}}_t,
\eel{cir}
where $a>0$, $\sigma>0$, $\kappa_{\mathbb P}\in {\mathbb R}$, and $W^{\mathbb{P}}$ is a one-dimensional standard Brownian motion (in this case $m=d=1$ and $n=0$). Consider the CIR pricing kernel in the form \eqref{affine_pk}. The short rate is given by \eqref{affineshortr} with $g=\gamma+au$ and $h=\delta-u\kappa_{\mathbb P}-u^2\sigma^2/2$. For simplicity we choose $\gamma=-au$ and $\delta=1+u\kappa_{\mathbb P}+u^2\sigma^2/2,$ so that the short rate can be identified with the state variable, $r_t=X_t$.
The market price of Brownian risk is $\lambda_t=\sigma u \sqrt{X_t}$. Under ${\mathbb Q}$ the short rate follows the process \eqref{affineq}, which  is again a CIR diffusion, but with a different rate of mean reversion:
\be
\kappa_{\mathbb Q}=\kappa_{\mathbb P}+\sigma^2u.
\ee

The fixed point $v$ of the Riccati ODE
$$\Psi'(t)=-\frac{1}{2}\sigma^2\Psi^2(t)-\kappa_{\mathbb P}\Psi(t)+\delta$$ with the initial condition
$\Psi(0)=u$
can be readily determined. Since $-\frac{1}{2}u^2\sigma^2-u\kappa_{\mathbb P}+\delta=1>0$, we know that $\Psi(0)=u$ is between the two roots of the quadratic equation $-\frac{1}{2}\sigma^2 x^2-\kappa_{\mathbb P} x+\delta=0$. This immediately implies that $\Psi(t)$ converges to the larger root, i.e.
\be
\lim_{t\rightarrow\infty}\Psi(t)=\frac{\sqrt{\kappa_{\mathbb P}^2+2\sigma^2\delta}-\kappa_{\mathbb P}}{\sigma^2}=\frac{\sqrt{\kappa_{\mathbb Q}^2+2\sigma^2}-\kappa_{\mathbb P}}{\sigma^2}=\frac{\kappa_{\mathbb L}-\kappa_{\mathbb P}}{\sigma^2}=:v,
\ee
where we introduce the following notation:
$$
\kappa_{\mathbb L}=\sqrt{\kappa_{\mathbb Q}^2+2\sigma^2}.
$$
Thus, the long bond in the CIR model is given by
\be
B_t^\infty=e^{\lambda t-\frac{\kappa_{\mathbb L}-\kappa_{\mathbb Q}}{\sigma^2}(X_t-X_0)}
\eel{cirlongbond}
with
\be
\lambda=\frac{a(\kappa_{\mathbb L}-\kappa_{\mathbb Q})}{\sigma^2}
\eel{cireigenvalue}
and the long bond volatility
$$
\sigma_t^\infty=-\frac{\kappa_{\mathbb L}-\kappa_{\mathbb Q}}{\sigma} \sqrt{X_t}.
$$
Under the long forward measure the state follows the process \eqref{affinel}, which  is again a CIR diffusion, but with the different rate of mean reversion $\kappa_{\mathbb L}>\kappa_{\mathbb Q}$. The fixed point $v$ is proportional to the difference between the rate of mean reversion under the long forward measure ${\mathbb L}$ and the data generating measure ${\mathbb P}$. It defines the market price of risk under ${\mathbb L}$ via $\lambda_t^\infty = v\sigma \sqrt{X_t}$.

We note that if one selects $u=(-\kappa_\mathbb{P}\pm\sqrt{\kappa_\mathbb{P}^2-2\sigma^2})/\sigma^2$ in the specification of the pricing kernel, then $v=0$ and $\lambda_t^\infty=0$, so the margingale component in the long term factorization is degenerate, and the pricing kernel is in the transition independent form. In this case, $\kappa_{\mathbb P}=\kappa_{\mathbb L}$ so that the data-generating measure coincides with the long-forward measure. This is the condition of Ross' recovery theorem (see \citet{linetsky_2014_cont} for more details).

Since the closed form solution for the CIR zero-coupon bond pricing function is available (\citet{cox_1985_2}), these results can also be recovered by directly calculating the limit
$$\lim_{T\rightarrow \infty}\frac{P(T-t,y)}{P(T,x)}=e^{\lambda t}\frac{\pi(y)}{\pi(x)}$$
with the eigenvalue $\lambda$ given by Eq.\eqref{cireigenvalue} and the eigenfunction $\pi(x)=e^{-\frac{\kappa_{\mathbb L}-\kappa_{\mathbb Q}}{\sigma^2}x}$.

\begin{remark}
\citet{borovicka_2014mis} in their Example 4 on p.2513 also consider an exponential-affine pricing kernel driven by a single CIR factor. However, their specification of the PK is in a special form such that $h=0$ in Eq.(4) for the short rate (which corresponds to the choice $\delta=u\kappa_{\mathbb P}+u^2\sigma^2/2$ in our parameterization). Thus, all dependence on the CIR factor is contained in the martingale component in the risk-neutral factorization of their PK, with the short rate being constant. In this special case the long bond is deterministic and the long forward measure is simply equal to the risk-neutral measure since the short rate is independent of the state variable.
In this special case the pricing operator has two distinct positive eigenfunctions. One of the eigenfunctions is constant. This eigenfunction defines the risk-neutral measure, which coincides with the long forward measure in this case due to independence of the short rate and the eigenfunction of the state variable. The second eigenfunction (Eq.(19) in \citet{borovicka_2014mis}) defines a probability measure, which is distinct from the risk-neutral measure and, hence, distinct from the long forward measure as well. Depending on the specific parameter values of the CIR process, either one of the two eigenfunctions may serve as the recurrent eigenfunction. The eigenmeasure associated with the other eigenfunction will not be recurrent, as the CIR process will have a non-mean reverting drift under that measure.
\end{remark}

\subsection{CIR Model with Absorption at Zero: ${\mathbb L}$ Exists, ${\mathbb Q}^{\pi_R}$ Does Not Exist}
\label{example_absorb}
We next consider a degenerate CIR model \eqref{cir} with $a=0$,
 $\sigma>0$, and $\kappa\in {\mathbb R}$. When $a$ vanishes, the diffusion has an absorbing boundary at zero, i.e. there is a positive probability to reach zero in finite time and, once reached, the process stays at zero with probability one for all subsequent times.
Consider a pricing kernel in the form of Eq.\eqref{affine_pk}. The short rate is given by \eqref{affineshortr} with $g=\gamma$ and $h=\delta-u\kappa_\mathbb{P}-\frac{1}{2}u^2\sigma^2$. We assume $\gamma=0$ and $\delta=1+u\kappa_\mathbb{P}+\frac{1}{2}u^2\sigma^2>0$, so that short rate $r_t$ takes values in $\mathbb{R}_+$. The market price of Brownian risk is $\lambda_t=\sigma u \sqrt{X_t}$, and under ${\mathbb Q}$ the short rate follows the process \eqref{affineq}, which  is again a CIR diffusion with an absorbing boundary at zero, but with a different rate of mean reversion $\kappa_\mathbb{Q}=\kappa_\mathbb{P}+\sigma^2 u$.

It is clear that under any locally equivalent measure, zero remains absorbing and thus no recurrent eigenfunction exists. Nevertheless, we can proceed in the same way as in our analysis of the CIR model to show that $$B_t^\infty=e^{-\frac{\kappa_\mathbb{L}-\kappa_\mathbb{Q}}{\sigma^2}(X_t-X_0)}$$ with $\kappa_\mathbb{L}=\sqrt{\kappa_\mathbb{Q}^2+2\sigma^2}$ is the long bond and $X_t$ solves the CIR SDE \eqref{cir} with $a=0$ and mean-reverting rate $\kappa_\mathbb{L}$ under ${\mathbb L}$. In fact, the treatment of the long bond and the long forward measure is exactly the same as in the non-degenerate example with $a>0$, even though this case is transient with absorption at zero. The eigenvalue degenerates in this case, $\lambda=0$, and the asymptotic long-term zero-coupon yield vanishes, corresponding to the eventual absorption of the short rate at zero.

\subsection{Vasicek Model}
\label{example_ou}
Our next example is the \citet{vasicek_1977equilibrium} model with the state variable following the OU diffusion:
\[
dX_t=\kappa(\theta_\mathbb{P}-X_t)dt+\sigma dW^{\mathbb{P}}_t
\]
with $\kappa>0$, $\sigma>0$ (in this case $m=0$, $n=d=1$).
Consider the pricing kernel in the form \eqref{affine_pk}.
The short rate is given by \eqref{affineshortr} with $g=\gamma+u\kappa\theta_\mathbb{P}-\frac{1}{2}u^2\sigma^2$ and
$h=\delta-u\kappa$. For simplicity we choose $\gamma=-u\kappa\theta_\mathbb{P}+\frac{1}{2}u^2\sigma^2$ and $\delta=1+u\kappa,$ so that the short rate is identified with the state variable, $r_t=X_t$.
The market price of Brownian risk is constant in this case, $\lambda_t=\sigma u$. Under ${\mathbb Q}$ the short rate follows the process \eqref{affineq}, which in this case is again the OU diffusion, but with a different long run mean $$\theta_\mathbb{Q}=\theta_\mathbb{P}-\frac{\sigma^2 u}{\kappa}$$ (the rate of mean reversion $\kappa$ remains the same).
The explicit solution to the ODE
$\Psi'(t)=-\kappa\Psi(t)+\delta$ with the initial condition $\Psi(0)=u$ is $$\Psi(t)=-(\frac{\delta}{\kappa}+u)e^{-\kappa t}+\frac{\delta}{\kappa},$$ and the limit yields the fixed point
$\lim_{t\rightarrow\infty}\Psi(t)=\frac{\delta}{\kappa}=:v.$
Thus, the long bond in the Vasicek model is given by
$$B_t^\infty=e^{\lambda t-\frac{1}{\kappa}(X_t-X_0)}$$ with the long-term yield
$$
\lambda=\theta_\mathbb{Q}-\frac{\sigma^2}{2\kappa^2}
$$
and the long bond volatility $$\sigma_t^\infty=-\frac{\sigma}{\kappa}.$$
Under the long forward measure the short rate follows the process \eqref{affinel}, which is again the OU diffusion, but with a different long run mean $$\theta_\mathbb{L}=\theta_\mathbb{Q}-\frac{\sigma^2}{\kappa^2}$$ (the rate of mean reversion remains the same).

\subsection{Non-mean-reverting Gaussian Model: $\mathbb{Q}^{\pi_R}$ Exists, $\mathbb{L}$ Does not Exist}\label{L_no_exist}

Suppose $X_t$ is a Gaussian diffusion with affine drift and constant volatility
\be
dX_t=\kappa(\theta-X_t)dt+\sigma dW^{\mathbb{P}}_t,
\ee
but now
with $\kappa<0$, so that the process is not mean-reverting. Consider a risk-neutral pricing kernel that discounts at the rate $r_t=X_t$, i.e. $S_t=e^{-\int_0^t X_s ds}$. Then the pure discount bond price is given by $P_t^T=P(X_t,T-t)$ with
\be
P(x,t)=A(t)e^{-x B(t)},
\ee
\be
B(t)=\frac{1-e^{-\kappa t}}{\kappa},\enskip A(t)=\exp\Big\{(\theta-\frac{\sigma^2}{2\kappa^2})(B(t)-t)-\frac{\sigma^2}{4\kappa}B^2(t)\Big\}.
\eel{bp_ou}
It is easy to see that the ratio $P(y,T-t)/P(x,T)$ does not have a finite limit as $T\rightarrow \infty$ and, hence, $P_t^T/P_0^T$ does not converge as $T\rightarrow \infty$.  Thus, the long bond and the long forward measure $\mathbb{L}$ do not exist in this case.
However, the recurrent eigenfunction $\pi_R$ and the recurrent eigen-measure $\mathbb{Q}^{\pi_R}$ do exist in this case and are explicitly given in Section 6.1.3 of  \citet{linetsky_2014_cont}. Under $\mathbb{Q}^{\pi_R}$, $X_t$ is the OU process with mean reversion (since $\kappa<0$):
\be
dX_t=(\sigma^2/\kappa-\kappa\theta+\kappa X_t)dt+\sigma dW_t^{\mathbb{Q}^{\pi_R}}.
\ee

\subsection{ Breeden Model}
Our next example is a special case of \citet{breeden_1979intertemporal} consumption CAPM  considered in Example 3.8 of \citet{hansen_2009}. There are two independent factors, a stochastic volatility factor $X_t^v$ evolving according to the CIR process
\be
dX_t^v=\kappa_v(\theta_v-X_t^v)dt+\sigma_v\sqrt{X_t^v} dW_t^{v,\mathbb{P}}
\ee
and a mean-reverting growth rate factor $X_t^g$ evolving according to the OU process
\[
dX_t^g=\kappa_g(\theta_g-X_t^g)dt+\sigma_g dW_t^{g,\mathbb{P}}.
\]
Here it is assumed that $\kappa_v,\kappa_g>0$, $\theta_v,\theta_g>0$, $\sigma_g>0$, $\sigma_v<0$ (so that a positive increment to $W^v$ reduces volatility), and $2\kappa_v\theta_v\geq \sigma_v^2$ (so that volatility stays strictly positive).
Suppose that equilibrium consumption evolves according to
\be
dc_t=X_t^g dt+\sqrt{X_t^v} dW_t^{v,\mathbb{P}}+\sigma_c dW_t^{g,\mathbb{P}},
\ee
where $c_t$ is the logarithm of consumption $C_t$. Thus, $X^g$ models predictability in the growth rate and $X^v$ models predictability in volatility.
Suppose also that the representative consumer's  preferences are given by
\be
\mathbb{E}\left[\int_0^\infty e^{-b t}\frac{C_t^{1-a}-1}{1-a}dt\right]
\ee
for $a,b>0$.
Then the implied pricing kernel $S_t$ is
\be
S_t=e^{-bt}C_t^{-a}=\exp\left(-a\int_0^t X_s^g ds-b t-a\int_0^t \sqrt{X_s^v} dW_s^{v,\mathbb{P}}-a\int_0^t \sigma_c dW_t^{g,\mathbb{P}}\right).
\ee
Using the SDEs for $X^g$ and $X^v$ it can be cast in the affine form \eqref{affine_pk}:
\be
\begin{array}{ll}
S_t & =\exp\left( -\gamma t-\frac{a}{\sigma_v}(X_t^v-X_0^v)-\frac{a\sigma_c}{\sigma_g}(X_t^g-X_0^g)\right. \\
 & \left.\quad-\frac{a\kappa_v}{\sigma_v}\int_0^t X_s^v ds-(a+\frac{a\sigma_c\kappa_g}{\sigma_g})\int_0^t X_s^gds\right),\\
\end{array}
\ee
where $\gamma=b-\frac{a\kappa_v\theta_v}{\sigma_v}-\frac{a\sigma_c\kappa_g\theta_g}{\sigma_g}$.
%By Proposition \ref{affine_ZCB} the $T$-maturity zero coupon bond valuation process is given by $P_t^T=e^{\Phi(T-t)+\Psi_1(T-t)X_t^v+\Psi_2(T-t)X_t^g}$, where
%\be
%\begin{split}
%& \Phi'(t)=\frac{1}{2}\sigma_g^2\Psi^2_1(t)+\kappa_v\theta_v \Psi_1(t)+\kappa_g\theta_g\Psi_2(t)+\frac{a\kappa_v\theta_v}{\sigma_v}+\frac{a\sigma_c\kappa_g\theta_g}{\sigma_g}-b,\enskip \Phi(0)=0,\\
%&\Psi_1'(t)=\frac{1}{2}\sigma_v^2\Psi^2_1(t)-\kappa_v\Psi_1(t)-\frac{a\kappa_v}{\sigma_v},\enskip \Psi_1(0)=-\frac{a}{\sigma_v},\\
%&\Psi_2'(t)=-\kappa_g\Psi_2(t)-a-\frac{a\sigma_c\kappa_g}{\sigma_g},\enskip \Psi_2(0)=-\frac{a\sigma_c}{\sigma_g}.\\
%\end{split}
%\ee
%Since the factors $X_t^v$ and $X_t^g$ are independent, the ODEs for $\Psi_1(t)$ and $\Psi_2(t)$ are decoupled. Thus we can check them individually.
\begin{proposition}
If $\kappa_g>0$ (mean-reverting growth rate) and $\kappa_v+\sqrt{\kappa_v^2+2a\kappa_v\sigma_v}+a\sigma_v>0$, Eq.\eqref{psi_converge} holds and, thus, Theorem \ref{affine_long} applies. The long bond is given by
\be
B_t^\infty=\exp\left(\lambda t+(\frac{a}{\sigma_v}-v_1)(X_t^v-X_0^v)+(\frac{a\sigma_c}{\sigma_g}-v_2)(X_t^g-X_0^g)\right),
\ee
where $\lambda=\gamma-\frac{1}{2}\sigma_g^2v_2^2+\kappa_v\theta_v v_1+\kappa_g\theta_g v_2$, $v_1=(\sqrt{\kappa_v^2+2a\kappa_v\sigma_v}-\kappa_v)/\sigma_v^2$,   $v_2=a(1/\kappa_g+\sigma_c/\sigma_g)$,
and the state variables have the following dynamics under ${\mathbb L}$:
\be dX_t^v=\left(\kappa_v\theta_v-\sqrt{\kappa_v^2+2a\kappa_v\sigma_v}X_t^v\right)dt+\sigma_v\sqrt{X_t^v}dW_t^{v,\mathbb{L}},
\ee
\be
dX_t^g=\kappa_g\left(\theta_g-\frac{a\sigma_g^2}{\kappa_g^2}-\frac{a\sigma_c\sigma_g}{\kappa_g}-X_t^g\right)dt+\sigma_g dW_t^{g,\mathbb{L}}.
\ee
\end{proposition}
\begin{proof}
In this model Eq.\eqref{riccati_d} reduces to
\be
\begin{split}
&\Phi^\prime(t)=-\frac{1}{2}\sigma_g^2\Psi_2(t)^2 +\kappa_v\theta_v\Psi_1(t)+\kappa_g\theta_g\Psi_2(t)+\gamma, \quad \Phi(0)=0,\\
&\Psi_1^\prime(t)=-\frac{1}{2}\sigma_v^2\Psi_1(t)^2 -\kappa_v\Psi_1(t)+\frac{a\kappa_v}{\sigma_v},\quad \Psi_1(0)=\frac{a}{\sigma_v},\\
&\Psi_2^\prime(t)=-\kappa_g\Psi_2(t)+a+\frac{a\sigma_c\kappa_g}{\sigma_g},\quad \Psi_2(0)=\frac{a\sigma_c}{\sigma_g}.\\
\end{split}
\ee
In this special case $\Psi_1(t)$ and $\Psi_2(t)$ are separated and thus can be analyzed independently. It is easy to see that if $\kappa_g>0$ then $\Psi_2(t)$ converges to $v_2$. When $\kappa_v+\sqrt{\kappa_v^2+2a\kappa_v\sigma_v}+a\sigma_v>0$, $\frac{a}{\sigma_v}$ is greater than the smaller root of the second order equation $-\frac{1}{2}\sigma_v^2\Psi_1(t)^2 -\kappa_v\Psi_1(t)+\frac{a\kappa_v}{\sigma_v}$, which implies that $\Psi_1(t)$ converges to the larger root of the second-order equation for $v_1$. The eigenvalue and the dynamics of the state variable can be computed accordingly. $\Box$.
\end{proof}

The proof essentially combines the proofs in Examples \ref{example_cir} and \ref{example_ou}. Similar to these examples, we observe that the rate of mean reversion of the volatility factor is higher under the long forward measure, $\sqrt{\kappa_v^2+2a\kappa_v\sigma_v}>\kappa_v$, while the rate of mean reversion of the growth rate remains the same, but its long run level is lower under ${\mathbb L}$.

\subsection{ \citet{borovicka_2014mis} Continuous-Time Long-Run Risks Model}

Our next example is a continuous-time version of the long-run risks model of \citet{bansal_2004risks} studied by
\citet{borovicka_2014mis}. It features growth rate predictability and stochastic volatility in the aggregate consumption and recursive preferences. The model is calibrated to the consumption dynamics in \citet{bansal_2004risks}. The two-dimensional state modeling growth rate predictability and stochastic volatility follows the affine dynamics:
{\small\be
d\begin{bmatrix}
X^1_t\\
X^2_t\\
\end{bmatrix}
=\left(
\begin{bmatrix}
0.013\\ 0
\end{bmatrix}+
\begin{bmatrix}
-0.013&0\\
0&-0.021\\
\end{bmatrix}
\begin{bmatrix}
X^1_t\\X^2_t
\end{bmatrix}\right)dt+\sqrt{X_t^1}
\begin{bmatrix}
-0.038&0\\
0&0.00034 \\
\end{bmatrix}
d
\begin{bmatrix}
W^{1,\mathbb{P}}_t\\
W^{2,\mathbb{P}}_t\\
\end{bmatrix},
\ee}
where $W^{i,\mathbb{P}}_t,$ $i=1,2,$ are two independent Brownian motions. Here $X^1_t$ is the stochastic volatility factor following a CIR process and $X^2_t$ is an OU-type mean-reverting growth rate factor with stochastic volatility. The aggregate consumption process $C_t$ in this model evolves according to
\be
d\log C_t=0.0015dt+X^2_tdt+\sqrt{X^1_t} 0.0078 dW^{3,\mathbb{P}}_t,
\ee
where $W^{3,\mathbb{P}}$ is a third independent Brownian motion  modeling direct shocks to consumption. Numerical parameters are from \citet{borovicka_2014mis} and are calibrated to monthly frequency (here time is measured in months). The representative agent in this model is endowed with recursive homothetic preferences and a unitary elasticity of substitution. \citet{borovicka_2014mis} solve for the pricing kernel:
\[
d\log S_t=-0.0035dt-0.0118X^1_tdt-X^2_t dt-\sqrt{X^1_t}\Big[0.0298\quad0.1330\quad0.0780\Big]dW^{\mathbb{P}}_t,
\]
where the three-dimensional Brownian motion $W^{\mathbb{P}}_t=(W^{i,\mathbb{P}}_t)_{i=1,2,3}$ is viewed as a column vector.

We now cast this model specification in the {\em three-dimensional} affine form of Assumption \ref{assumption_affine_PK}. To this end, we introduce a third factor
$X^3_t=\log S_t$. We can then write the pricing kernel in the exponential affine form $S_t=e^{X_t^3}$, where the state vector $(X^1_t,X^2_t,X^3_t)$ follows a three-dimensional affine diffusion driven by a three-dimensional Brownian motion:
\be
dX_t
=\left(
b+
B
X_t\right)dt+\sqrt{X_t^1}
\rho
dW^{\mathbb{P}}_t,
\ee
where the numerical values for entries of the three-dimensional vector $b$ and $3\times 3$-matrices $B$ and $\rho$ are given above.

We can now directly apply our general results for affine pricing kernels. First,
by Theorem \ref{RN_affine}, the short rate is $r(X_t)=0.0035-0.00057798 X^1_t+X^2_t$ and depends only on the factors $X^1$ and $X^2$ and is independent of $X^3$. The risk-neutral (${\mathbb Q}$-measure) dynamics is given by:
\be
d\begin{bmatrix}
X^1_t\\
X^2_t\\
X^3_t\\
\end{bmatrix}
=\left(
\begin{bmatrix}
0.013\\0\\-0.0035
\end{bmatrix}+
\begin{bmatrix}
-0.0119&0&0\\
-0.00004522&-0.021&0\\
0.0129&-1&0\\
\end{bmatrix}
\begin{bmatrix}
X^1_t\\X^2_t\\X^3_t
\end{bmatrix}\right)dt+\sqrt{X_t^1}\rho dW^{\mathbb{Q}}_t,
\ee
where
\be
\rho=\begin{bmatrix} -0.038&0&0\\0&0.00034&0\\-0.0298&-0.1330&-0.0780\\
\end{bmatrix}.
\ee
%The pure discount bonds are given by Proposition \ref{affine_ZCB}:
%\be
%P_t^T=\mathbb{E}^{\mathbb P}[e^{X^3_T-X^3_t}|\mathscr{F}_t]=e^{\Phi(T-t)+\Psi_1(T-t) X_t^1+\Psi_2(T-t)X_t^2+(\Psi_3(T-t)-1) X^3_t}.
%\ee
%where  $\Phi(t)$ and
The vector $\Psi(t)=(\Psi_1(t),\Psi_2(t),\Psi_3(t))^\dagger$ solves the ODE (here $\alpha:=\rho\rho^\dagger$):
$$
\Psi_1^\prime(t)=-\frac{1}{2}\Psi(t)^\dagger \alpha\Psi(t)+B_{11}\Psi_1(t)+B_{21}\Psi_2(t)+B_{31}\Psi_3(t),
$$
$$
\Psi_2^\prime(t)=B_{22}\Psi_2(t)+B_{32}\Psi_3(t),\quad
\Psi_3^\prime(t)=0
$$
with $\Phi(0)=\Psi_1(0)=\Psi_2(0)=0, \Psi_3(0)=-1$. It is immediate that
$$\Psi_3(t)\equiv-1\quad \text{and}\quad \Psi_2(t)=\frac{B_{32}}{B_{22}}(1-e^{B_{22}t})$$ and, since $B_{22}<0$, $$\lim_{t\rightarrow\infty}\Psi_2(t)=B_{32}/B_{22}=47.6191:=v_2.$$
To see $\Psi_1(t)$ convergence, notice that we can write $-\frac{1}{2}\Psi(t)^\dagger \alpha\Psi(t)+B_{11}\Psi_1(t)+B_{21}\Psi_2(t)+B_{31}\Psi_3(t)=c_1(\Psi_1(t))^2 + c_2 \Psi_1(t) + c_3 (\Psi_2(t))^2 + c_4 \Psi_2(t) + c_5$, where $c_1, c_2, c_3, c_4, c_5<0$. Since $\Psi_1(0)=\Psi_2(0)=0$, we have $\Psi_1'(0)<0$. Since $\Psi_2(t)>0$ and it is easy to see that $\Psi_1(t)<0$. Since $\Psi_2(t)<v_2$, we have $c_1(\Psi_1(t))^2 + c_2 \Psi_1(t) + c_3 (\Psi_2(t))^2 + c_4 \Psi_2(t) + c_5>c_1(\Psi_1(t))^2 + c_2 \Psi_1(t) + c_3 v_2^2+c_4 v_2+c_5$. We can check that $c_1(\Psi_1(t))^2 + c_2 \Psi_1(t) + c_3 v_2^2+c_4 v_2+c_5=0$ has two negative roots. Denote the larger root $v_1$, we see that $\Psi_1(t)>v_1$. Combining these facts, we see that $\Psi_1(t)$ converges to $v_1$. The exact value of $v_1$ has to be determined numerically. The numerical solution yields $$v_1=\lim_{t\rightarrow\infty}\Psi_1(t)=-0.2449.$$
%Figure \ref{plot_psi} plots the functions $\Psi_1(t)$ and $\Psi_2(t)$.
In Figure \ref{phipsi}, we plot the functions $\Psi_1(t)$ and $\Psi_2(t)$, as well as the gross return $B_t^{t+T}$ on the $T$-bond over the period $[0,t]$ as a function of $T$. In this numerical example we take $t=12$ months, so we are looking at the one-year holding period return, and assume that the initial state $X_0$ and the state $X_t$ are both equal to the stationary mean under $\mathbb{P}$. We observe that in this model specification $\Psi(t)$ and $B_t^{t+T}$ are already very close to the fixed point for $t$ around 30 years (360 months).
\begin{figure}[h]
\includegraphics[width=60mm, height=50mm]{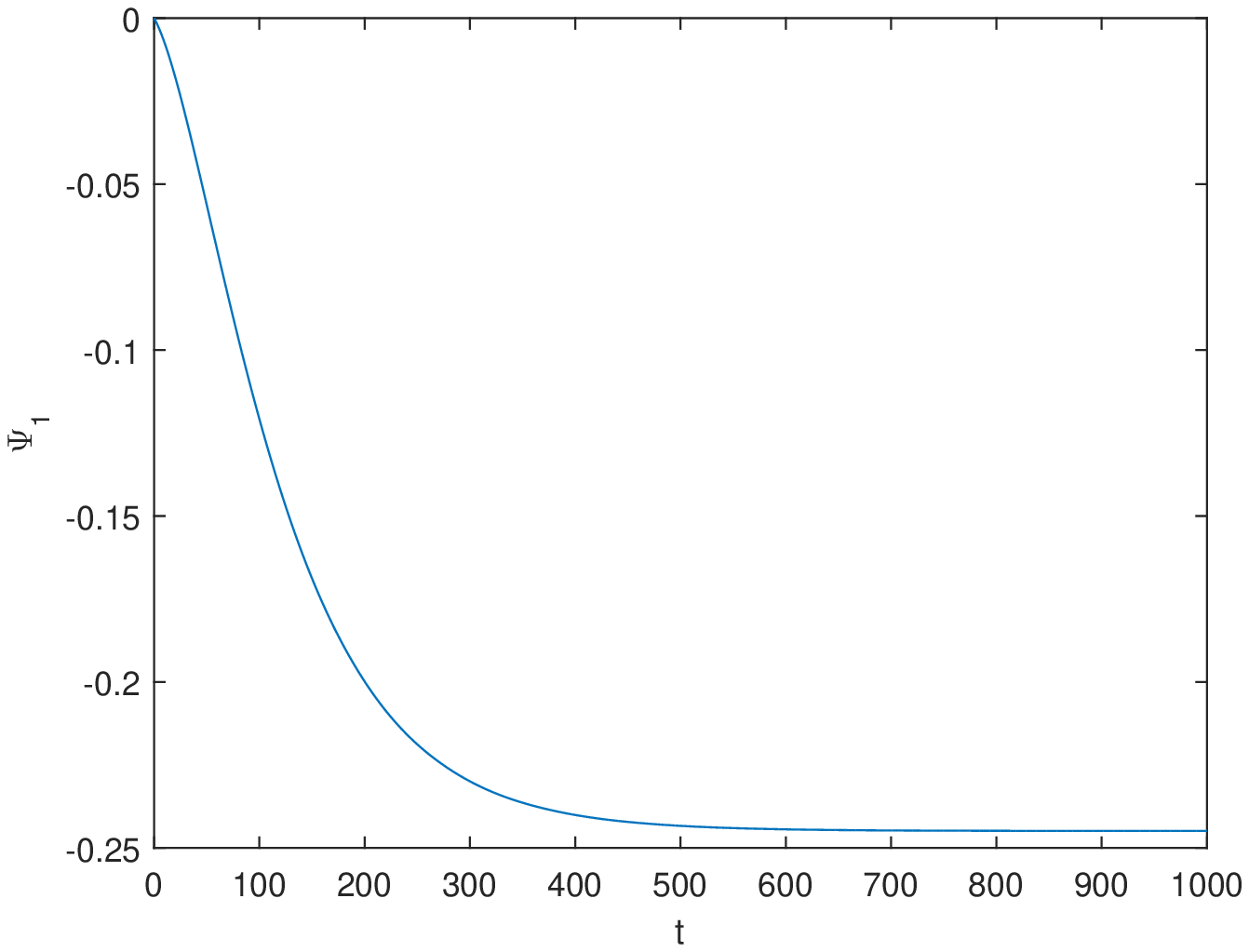}
\includegraphics[width=60mm, height=50mm]{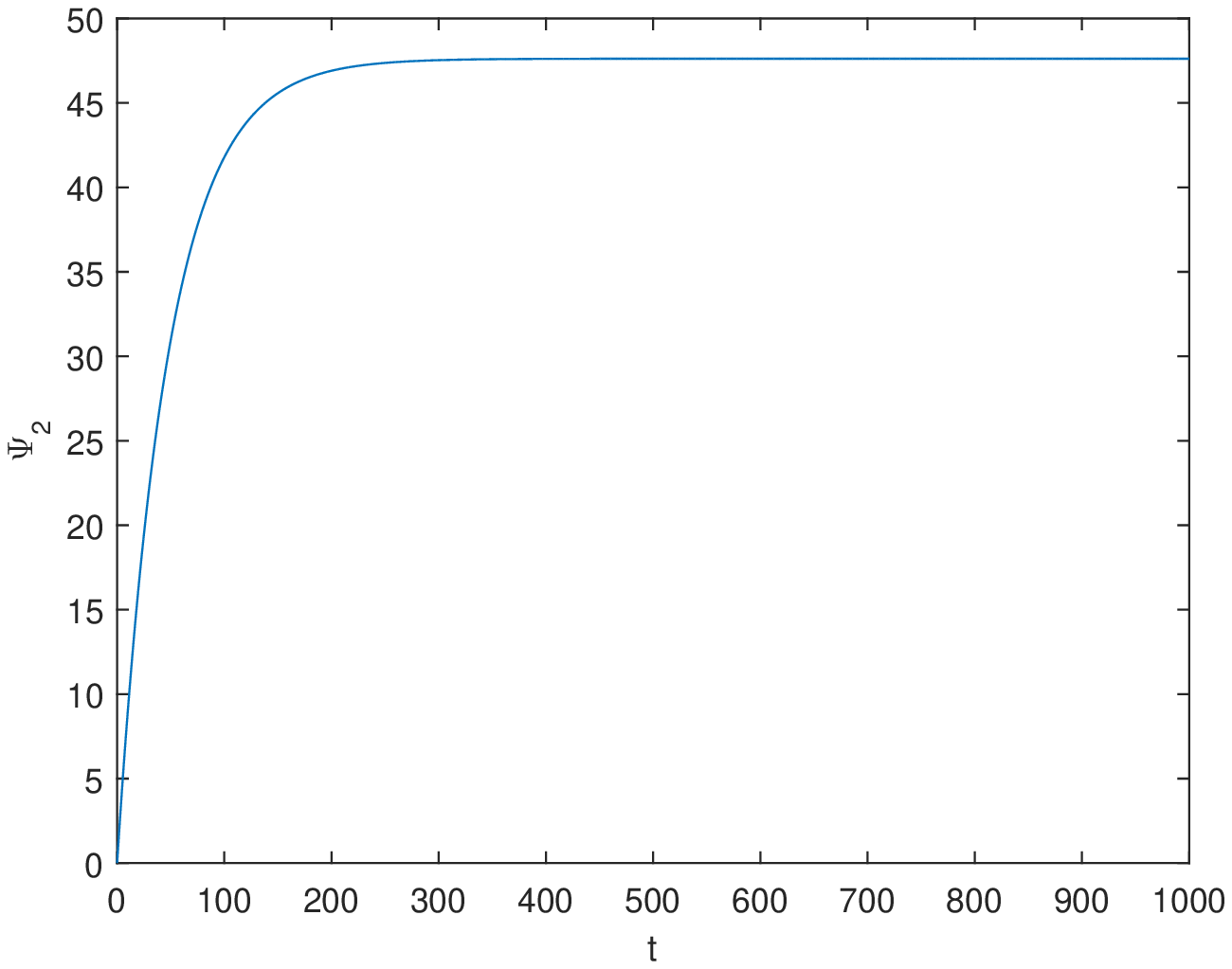}
\includegraphics[width=60mm, height=50mm]{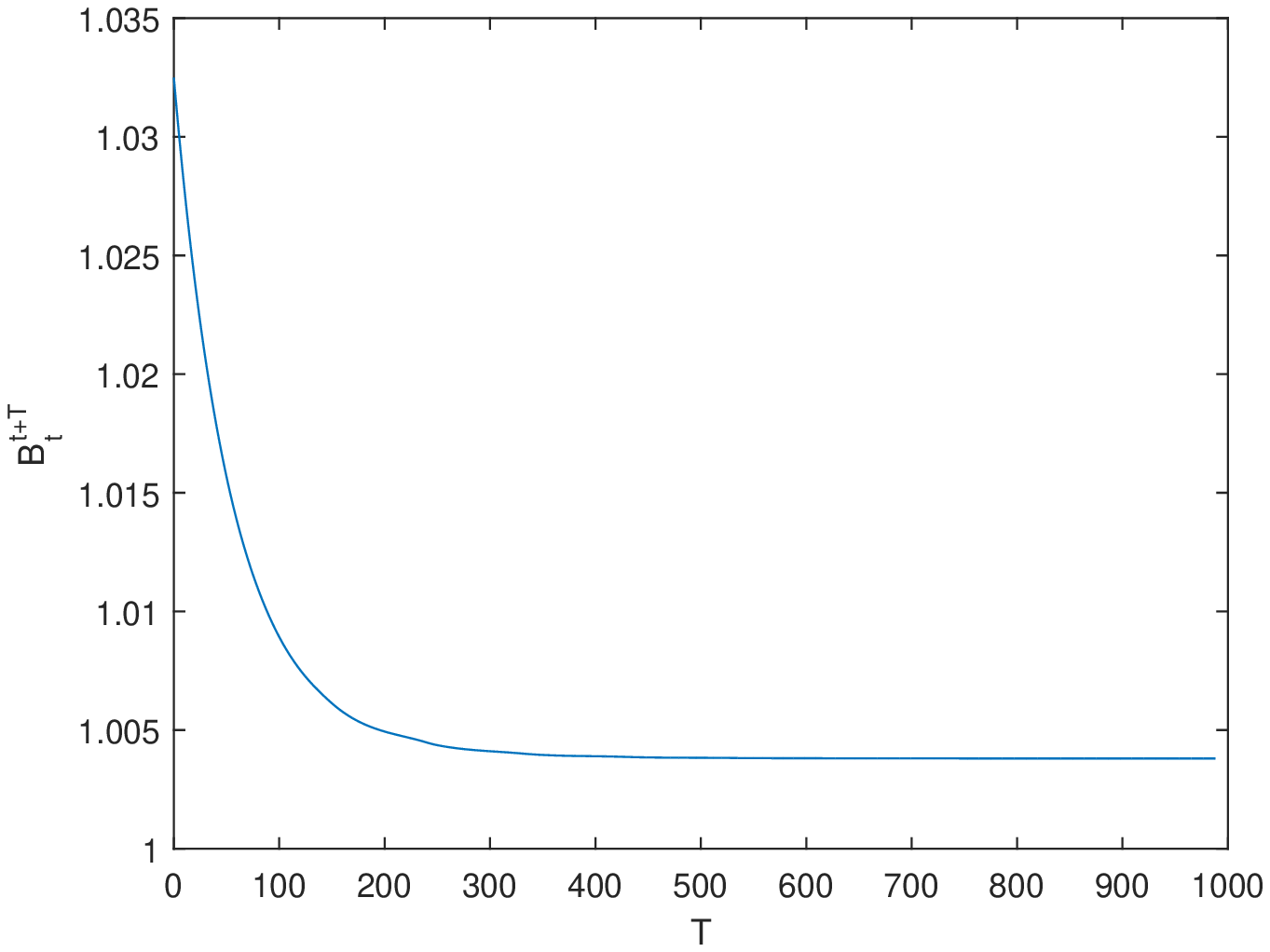}
\caption{Plot of $\Psi_1(t)$, $\Psi_2(t)$ and $B_t^{t+T}$. Time is measured in months.}
\label{phipsi}
\end{figure}

By Theorem \ref{affine_long}, the eigenfunction determining the long bond is $\pi(x)=e^{-v_1 x^1-v_2 x^2},$ corresponding to the eigenvalue (note this is not annualized yield since time unit is in month)
$$\lambda=b_1v_1+b_2v_2-b_3=0.0003163,$$ the long bond is given by $$B_t^\infty=e^{\lambda t - v_1(X_t^1-X_0^1)-v_2(X_t^2-X_0^2)},$$ the  martingale component is given by $$M_t^\infty=e^{\lambda t - v_1(X_t^1-X_0^1)-v_2(X_t^2-X_0^2)+X_t^3},$$ and  the state vector $(X^1_t, X_t^2, X^3_t)$ has the following dynamics under the long forward measure ${\mathbb L}$:
\be
d\begin{bmatrix}
X^1_t\\
X^2_t\\
X^3_t\\
\end{bmatrix}
=\left(
\begin{bmatrix}
0.013\\0\\-0.0035
\end{bmatrix}+
\begin{bmatrix}
-0.0115&0&0\\
-0.00005074&-0.021&0\\
0.0153&-1&0\\
\end{bmatrix}
\begin{bmatrix}
X^1_t\\X^2_t\\X^3_t
\end{bmatrix}\right)dt+\sqrt{X_t^1}\rho dW^{\mathbb{L}}_t.
\ee

As already observed by \citet{borovicka_2014mis}, in this model the state dynamics under the long forward measure ${\mathbb L}$ is close to the state dynamics under the risk-neutral measure ${\mathbb Q}$ and is substantially distinct from the dynamics under the data-generating measure ${\mathbb P}$ due to the volatile martingale component $M_t^\infty$.
However, our approach to the analysis of this model is different from  the analysis of \citet{borovicka_2014mis}. We cast it as a three-factor affine model and directly apply our Theorem \ref{affine_long} for affine models that is, in turn, a consequence of our Theorem \ref{implication_L1} for semimartingale models. We only need to determine the fixed point \eqref{psi_converge} of the Riccati equation. Existence of the long bond, the long term factorization of the pricing kernel, and the long forward measure then immediately follow from Theorem \ref{affine_long}, without any need to verify ergodicity.
In fact, the three-factor affine process $(X^1_t,X_t^2,X_t^3)$ is not ergodic, and not even recurrent, as is immediately seen from the dynamics of $X^3$.
In contrast, the approach in \citet{borovicka_2014mis} relies on the two-dimensional mean-reverting affine diffusion $(X^1_t,X_t^2)$. Namely, since the Perron-Frobenius theory of \citet{hansen_2009} requires ergodicity to single out the principal eigenfunction and ascertain its relevance to the long-term factorization,  \citet{borovicka_2014mis} implicitly split the pricing kernel into the product of two sub-kernels, a multiplicative functional of the two-dimensional Markov process $(X^1_t,X_t^2)$ and the additional factor in the form $e^{-\int_0^t 0.0780\sqrt{X_s^1}dW_s^{3,{\mathbb P}}}$. The Perron-Frobenius theory of \citet{hansen_2009} is then applied to the multiplicative functional of the two-dimensional Markov process $(X^1_t,X_t^2)$.
In contrast, in our approach we do not require ergodicity and work directly with the non-ergodic three-dimensional process and verify that the Riccati ODE possesses a fixed point, which is already sufficient for existence of  the long-term factorization in affine models by Theorem \ref{affine_long}.

\section{Conclusion}

This paper constructs and studies the long-term factorization of affine pricing kernels into discounting at the rate of return on the long bond and the martingale component that accomplishes the change of probability measure to the long forward measure. It is shown that the principal eigenfunction of the affine pricing kernel germane to the long-term factorization is an exponential-affine function of the state vector with the coefficient vector identified with the fixed point of the Riccati ODE. The long bond volatility and the volatility of the martingale component are explicitly identified in terms of this fixed point. When analyzing a given affine model, a research needs to establish whether the Riccati ODE possesses a fixed point. If the fixed point is determined, the long-term factorization then follows. It is shown how the long-term factorization plays out in a variety of asset pricing models, including single factor CIR and Vasicek models, a two-factor version of Breeden's CCAPM, and the three-factor long-run risks model studied in \citet{borovicka_2014mis}.

\bibliography{mybib7}

\end{document}